\documentclass[conference]{IEEEtran}
\IEEEoverridecommandlockouts
\usepackage{cite}
\usepackage{amsmath,amssymb,amsfonts}
\usepackage{graphicx}
\usepackage{comment}
\usepackage{textcomp}
\usepackage{booktabs} 
\usepackage{multirow} 
\usepackage[caption=false,font=footnotesize]{subfig}
\usepackage{url}
\usepackage{mathtools}
\usepackage{algorithm} 
\usepackage{algpseudocode}
\usepackage{float}
\usepackage{amsthm}
\newtheorem{theorem}{Theorem}
\usepackage{xcolor}
\def\BibTeX{{\rm B\kern-.05em{\sc i\kern-.025em b}\kern-.08em
    T\kern-.1667em\lower.7ex\hbox{E}\kern-.125emX}}
\begin{document}
\raggedbottom

\title{VAMP-Diff: VampPrior Latent Diffusion 
for Photoplethysmography Modeling%
\thanks{\footnotesize Code at: 
\texttt{https://github.com/fatighasemi74/VAMP-Diff}}}

\author{
\IEEEauthorblockN{
Fatemeh Ghasemi Balouei$^{1}$ \quad
Nathan Willemsen$^{2}$ \quad
Mahesh Banavar$^{1}$ \quad
Bahman Moraffah$^{3}$
}
\IEEEauthorblockA{
$^{1}$Department of Electrical and Computer Engineering, Clarkson University, Potsdam, USA\\
$^{2}$Data Science, Worcester Polytechnic Institute, Worcester, USA\\
$^{3}$Department of Computer Science, Data Science, and Artificial Intelligence, Worcester Polytechnic Institute, Worcester, USA\\
\{ghasemf, mbanavar\}@clarkson.edu \quad
\{nawillemsen, bmoraffah\}@wpi.edu
}
}
\maketitle

\begin{abstract}
Photoplethysmography (PPG) has become a ubiquitous physiological signal; however, current generative models still struggle to preserve realistic waveform morphology and learn a latent structure that captures cardiac and respiratory physiology. PPG generators trained with adversarial losses can produce plausible waveforms, but provide no inference path from a real signal to a latent representation. Variational autoencoders, on the other hand, map the PPG data to latent codes, although their decoders often blur systolic upstrokes and dampen amplitude and spectral details. Diffusion models improve waveform fidelity, but typically lack an inference path for reconstruction and physiological analysis. We propose VampPrior Latent Diffusion (VAMP-Diff), a jointly trained variational diffusion model that combines a temporal PPG encoder, a conditional one-dimensional diffusion decoder, and VampPrior regularization on a compact pooled latent. The model uses full temporal latent during diffusion reconstruction, giving the decoder access to beat timing and morphology, while generating samples from learned VampPrior components instead of a fixed Gaussian prior. We demonstrate on the CapnoBase dataset that VAMP-Diff produces realistic PPG signals, reconstructs sharper physiological waveforms than Gaussian-prior baselines, preserves heart-rate information, maintains respiratory-rate consistency, and is sensitive to waveform corruptions through reconstruction error.
\end{abstract}


\section{Introduction}
\label{sec:intro}

Photoplethysmography (PPG)~\cite{allen2007photoplethysmography, elgendi2012ppg, charlton2022wearable} is a non-invasive optical technique that encodes cardiac and respiratory physiology in a periodic waveform. Generative modeling of PPG requires preserving both global structure, such as dominant periodicity and amplitude, and local morphology, such as the systolic upstroke, dicrotic notch, and inter-beat variation. This is a challenging task since clinically relevant information in PPG signals is distributed across multiple timescales: cardiac periodicity at $0.7$--$3.0$ Hz, respiratory modulation at $0.1$--$0.5$ Hz, and beat-level morphological features that occur within individual cycles. A generative model that collapses this temporal structure into a single latent vector loses the positional information needed to reconstruct these features faithfully. The scarcity of large labeled PPG datasets motivates the need for generative models that produce physiologically realistic synthetic signals for downstream tasks such as respiratory rate estimation~\cite{bian2020respiratory, aqajari2021end}. Generative approaches for PPG and related physiological signals include generative adversarial networks (GANs)~\cite{goodfellow2014generativeadversarialnetworks, article, vo2021p2e, NEURIPS2019_c9efe5f2} that optimize adversarial objectives to match signal statistics and diffusion models~\cite{ho2020denoising, sohl2015deep, alcaraz2023diffusion, tashiro2021csdiconditionalscorebaseddiffusion} that model the full conditional distribution through iterative denoising. Latent diffusion models~\cite{Rombach_2022_CVPR,preechakul2022diffusion, vahdat2021scorebasedgenerativemodelinglatent} combine a learned latent space with a diffusion decoder but train the two components separately. Variational 
autoencoders (VAEs)~\cite{kingma2013auto} learn structured latent representations, and hierarchical VAEs~\cite{NEURIPS2020_e3b21256, child2020very} have demonstrated strong performance on image generation. Diffusion priors have also been introduced into VAEs to improve generation~\cite{wehenkel2021diffusion, kingma2021variational}, Recent work has also combined the VampPrior~\cite{tomczak2018vampprior} with hierarchical VAEs to improve prior expressiveness~\cite{kuzina2024hierarchical}; however, this occurs in a setting different from one-dimensional physiological signal generation. None of these approaches jointly trains a temporal encoder and a diffusion decoder under a data-adaptive prior. Without joint training, the latent space cannot adapt to both the KL regularization and the decoder's reconstruction objective simultaneously, which limits the fidelity of generated waveforms and their utility for physiological inference.

We propose VAMP-Diff, a jointly trained variational diffusion model that combines a temporal PPG encoder, a conditional one-dimensional diffusion decoder, and a VampPrior defined on a compact pooled latent, evaluated on CapnoBase~\cite{karlen2021capnobase}. We also introduce a morphology-preserving loss that supplements the variational diffusion objective with spectral, derivative, amplitude, and peak-to-peak terms to enforce waveform fidelity beyond sample-level error. VAMP-Diff reconstructs sharper physiological waveforms than Gaussian-prior baselines, preserves heart-rate information, and maintains respiratory-rate estimator consistency on held-out test signals. 

The remainder of the paper is organized as follows. Section~\ref{sec:background} provides background on PPG, VAEs, and diffusion models. Section~\ref{sec:hybridmodel} describes the proposed model and training objective. Section~\ref{sec:downstream} defines downstream evaluation tasks, and Section~\ref{sec:empirical_results} reports empirical results. Finally, concluding remarks are provided in Section~\ref{sec:conclusions}. 

\section{Background}
\label{sec:background}

\subsection{Photoplethysmography}
\label{subsec:ppg_background}

Photoplethysmography (PPG) is a non-invasive optical technique that measures pulsatile changes in blood volume in peripheral tissue, producing a periodic waveform that encodes heart rate, respiratory rate, and autonomic activity~\cite{allen2007photoplethysmography, elgendi2012ppg}. Since PPG sensors are inexpensive and easy to deploy, they are widely used in clinical monitors and wearable devices, with applications that range from cardiovascular monitoring to biometric analysis~\cite{Lovisotto_2020, venkataswamy2025descriptor}. However, PPG waveforms are difficult to model since clinically relevant information is distributed across peak timing, systolic upstroke, dicrotic morphology, amplitude changes, and slower respiratory variation across the window. Large labeled PPG datasets are also hard to obtain because data collection often requires controlled acquisition, expert annotation, and privacy protection under regulations such as HIPAA, which leaves only a small number of public datasets such as CapnoBase and related physiological signal resources~\cite{karlen2021capnobase, 7748483, Lovisotto_2020, venkataswamy2025descriptor}.

\subsection{Variational Autoencoders and VampPrior}
\label{subsec:vae_vampprior_background}

Variational Autoencoders (VAEs)~\cite{kingma2013auto} learn an encoder \(q_\phi(z\mid x_0)\), a decoder \(p_\theta(x_0\mid z)\), and a prior \(p(z)\) through the evidence lower bound~\cite{odaibo2019tutorial}
\begin{equation}
    \mathcal{L}_{\mathrm{VAE}}
    =
    \mathbb{E}_{q_\phi(z\mid x_0)}
    \!\left[
    \log p_\theta(x_0\mid z)
    \right]
    -
    D_{\mathrm{KL}}
    \!\left(
    q_\phi(z\mid x_0)
    \,\|\,p(z)
    \right).
    \label{eq:vae_elbo}
\end{equation}
The reconstruction loss trains the decoder to recover \(x_0\) from the latent variable, and the KL penalty regularizes the encoder posterior against the generation prior. A vanilla VAE usually assumes a fixed prior \(p(z)=\mathcal{N}(0, I)\), which can poorly match the aggregate posterior when the data lies on a structured or curved latent manifold. This mismatch can degrade generation, since samples from \(\mathcal{N}(0, I)\) may fall outside the latent regions that guide decoder training.

The VampPrior~\cite{tomczak2018vampprior} addresses this mismatch by defining the prior as a mixture of encoder posteriors evaluated at \(K\) learnable pseudo-inputs \(\{u_k\}_{k=1}^{K}\),
\begin{equation}
    p_\psi(z)
    =
    \frac{1}{K}
    \sum_{k=1}^{K}
    q_\phi(z\mid u_k).
    \label{eq:vampprior}
\end{equation}
The encoder produces each mixture component, which allows the prior to adapt to the geometry of the learned posterior during training. The corresponding VAE objective then is 
\begin{equation}
    \mathcal{L}_{\mathrm{Vamp}}
    =
    \mathbb{E}_{q_\phi(z\mid x_0)}
    \!\left[
    \log p_\theta(x_0\mid z)
    \right]
    -
    D_{\mathrm{KL}}
    \!\left(
    q_\phi(z\mid x_0)
    \,\|\,p_\psi(z)
    \right).
    \label{eq:vamp_elbo}
\end{equation}
The compact VampPrior regularizer in Section~\ref{subsec:encoder} builds on this prior and applies the KL to a pooled temporal latent.

\subsection{Denoising Diffusion Models}
\label{subsec:diffusion}

Denoising diffusion probabilistic models (DDPMs)~\cite{ho2020denoising} generate data by learning to reverse a gradual Gaussian noising process. Starting from a clean signal \(x_0\), the forward process samples
\begin{equation}
    q(x_t\mid x_0)
    =
    \mathcal{N}
    \!\left(
    x_t;
    \sqrt{\bar{\alpha}_t}\,x_0,
    (1-\bar{\alpha}_t)I
    \right),
    \qquad
    t=1,\ldots,T,
    \label{eq:ddpm_forward_background}
\end{equation}
where \(\bar{\alpha}_t\) is the cumulative noise schedule. A neural network, often a U-Net architecture~\cite{DBLP:journals/corr/RonnebergerFB15}, is trained to recover the clean signal or the injected noise from \(x_t\). In this paper, our decoder uses the clean-signal prediction parameterization \(\widehat{x}_0\), since the losses in Section~\ref{subsec:loss} act directly on the reconstructed waveform. Sampling from a diffusion model can require many reverse steps, so DDIM~\cite{song2020denoising} is used at inference to obtain a deterministic and faster reverse trajectory.
\section{VAMP-Diff: VampPrior Latent Diffusion for PPG}
\label{sec:hybridmodel}

\subsection{Temporal Latent Encoder and Prior Regularization}
\label{subsec:encoder}

The encoder maps each PPG window \(x_0\in\mathbb{R}^{L}\) to a diagonal Gaussian posterior over a temporal latent sequence,
\begin{equation}
  q_\phi(z\mid x_0)
=
\mathcal N\!\left(
z;\mu_\phi(x_0),\operatorname{diag}(\sigma_\phi^2(x_0))
\right),
\qquad
z\in\mathbb{R}^{C_z\times T_z}.  
\end{equation}
The sequence structure of \(z\) is retained so that the decoder receives temporally ordered information across the PPG window. A single pooled code can remove the temporal ordering before the denoising decoder receives the latent representation. For unconditional generation, the standard Gaussian prior \(p(z)=\mathcal N(0,I)\) does not match the aggregate posterior learned by the encoder, thus the VAMP-Diff employs the VampPrior~\eqref{eq:vampprior}. VampPrior is built from \(K\) learnable pseudo-inputs \(\{u_k\}_{k=1}^{K}\), each passed through the same encoder \(q_\phi(z\mid x_0)\). The pseudo-input initialization samples real training windows under heart-rate and amplitude stratification, as random pseudo-input initialization may lead to an unstable KL behavior early in training. Since \(z\in\mathbb{R}^{C_z\times T_z}\) is high-dimensional, we evaluate the VampPrior KL in a pooled latent space. Let
\[
A:\mathbb{R}^{C_z\times T_z}\to\mathbb{R}^{C_z\times T_c},
\qquad T_c<T_z,
\]
denote temporal pooling, and define \(\widetilde z=Az\). During training, the compact posterior
\[
q_\phi^A(\widetilde z\mid x_0)=A_{\#}q_\phi(z\mid x_0)
\]
is matched to the compact VampPrior \(p_\psi^A(\widetilde z)\). For generation, the decoder does not sample from the pooled space; it receives a full temporal latent sampled from the pseudo-input posterior \(q_\phi(z\mid u_k)\).

\subsection{Conditional Diffusion Decoder}
\label{sec:decoder}

The decoder represents \(p_\theta(x_0\mid z)\) through a conditional one-dimensional diffusion U-Net that reconstructs a clean PPG waveform from a noisy diffusion state and the full temporal latent. Given the forward diffusion process
\begin{equation}
\label{eq:diff_forward}
\begin{aligned}
x_t
&=
\sqrt{\bar\alpha_t}\,x_0
+
\sqrt{1-\bar\alpha_t}\,\epsilon,
\\
\epsilon
&\sim \mathcal N(0,I),
\qquad
t\sim\mathrm{Unif}\{1,\ldots,T\}.
\end{aligned}
\end{equation}

the network predicts the clean signal directly as \(\widehat{x}_0
=
f_\theta(x_t,t,z).\) The decoder receives \(z\in\mathbb{R}^{C_z\times T_z}\) at full temporal resolution through multiscale FiLM conditioning~\cite{perez2018film}. At decoder scale \(\ell\), the latent is temporally matched to the feature resolution and projected to affine FiLM parameters,
\begin{equation}
   (\gamma^{(\ell)}(z),\beta^{(\ell)}(z))
=
g_\ell\!\left(A_\ell z\right), 
\end{equation}
where \(A_\ell\) is temporal resampling to the resolution of scale \(\ell\), and \(g_\ell\) is a learned projection network. For decoder features \(h^{(\ell)}\), the conditioned feature map is then 
\begin{equation}
    \mathrm{FiLM}_{\ell}(h^{(\ell)},z)
=
\gamma^{(\ell)}(z)\odot h^{(\ell)}
+
\beta^{(\ell)}(z).
\end{equation}

We jointly train the encoder, VampPrior, and diffusion decoder in a single objective. The latent representation therefore receives reconstruction gradients from the decoder and prior-regularization gradients from the VampPrior during the same optimization, whereas the latent diffusion models first train a compressor and then freeze it for diffusion training \cite{Rombach_2022_CVPR}.

\paragraph{Inference with DDIM}
At inference, we use DDIM sampling with \(\eta=0\), which defines a deterministic decoding trajectory for a fixed latent \(z\) and fixed initial noise \(x_T\) \cite{song2020denoising}.  Reconstruction conditions the decoder on the encoder-posterior latent, whereas unconditional generation conditions it on a full-resolution latent sampled from the VampPrior.

\subsection{Full-Resolution Latent Training Objective}
\label{subsec:objective}

The VAMP-Diff architecture imposes two coupled training requirements in which the encoder learns a temporal latent representation that reconstructs each observed PPG window through the conditional diffusion decoder. The same latent space is then regularized so that samples drawn from the learned VampPrior decode into physiological waveforms at generation time. For a window \(x_0\in\mathbb{R}^{L}\), the encoder defines \(q_\phi(z\mid x_0)\) over the sequence latent \(z\in\mathbb{R}^{C_z\times T_z}\), and the diffusion decoder predicts \(\widehat{x}_0=f_\theta(x_t,t,z)\) from the noisy signal \(x_t\), timestep \(t\) and full latent \(z\). The VampPrior KL is evaluated on the pooled latent \(\widetilde z\) and \(T_c<T_z\), since direct mixture-prior matching in the full sequence latent space is computationally expensive and unstable. Using \(z\) for decoding and \(\widetilde z\) for the VampPrior KL preserves temporal detail for waveform reconstruction and keeps prior matching in a compact latent space.

\begin{theorem}
\label{thm:vampdiff_objective}
Let \(x_0 \in \mathbb{R}^{L}\) be a PPG window drawn from \(p_{\mathrm{data}}\). 
The encoder defines a diagonal Gaussian posterior over a temporal latent sequence
\[
q_\phi(z\mid x_0)
=
\mathcal N\!\left(
z;\mu_\phi(x_0),\operatorname{diag}(\sigma_\phi^2(x_0))
\right),
\qquad
z\in \mathbb{R}^{C_z\times T_z}.
\]
Suppose \(A:\mathbb{R}^{C_z\times T_z}\to \mathbb{R}^{C_z\times T_c}\), with \(T_c<T_z\) is the temporal pooling operator used for KL evaluation, and define \(\widetilde z =: A z .\) The induced compact posterior is
\begin{equation}
  q_\phi^{A}(\widetilde z\mid x_0)
=
A_{\#}q_\phi(z\mid x_0),  
\end{equation}
where \(A_{\#}\) is the pushforward distribution under \(A\). 
Given learnable pseudo-inputs \(\{u_k\}_{k=1}^{K}\), the compact VampPrior is \(p_\psi^{A}(\widetilde z)
=
\frac{1}{K}\sum_{k=1}^{K}
q_\phi^{A}(\widetilde z\mid u_k)\).

For the diffusion forward process Equation~\eqref{eq:diff_forward}, let \(f_\theta(x_t,t,z)\) be a conditional diffusion decoder that predicts the clean signal \(x_0\). 
Then the VAMP-Diff training loss is
\begin{equation}
\label{eq:vamp_diff}
\begin{aligned}
\mathcal L_{\mathrm{VAMP\text{-}Diff}}(\theta,\phi,\psi)
&=
\mathbb E_{\substack{
x_0\sim p_{\mathrm{data}}\\
z\sim q_\phi(z\mid x_0)\\
t\sim \mathrm{Unif}\{1,\ldots,T\}\\
\epsilon\sim\mathcal N(0,I)
}}
\left[
\left\|
x_0
-
f_\theta(x_t,t,z)
\right\|_2^2
\right]
\\
&\hspace{-5.0em}
+
\beta\,
\mathbb E_{x_0\sim p_{\mathrm{data}}}
...
\left[
D_{\mathrm{KL}}\!\left(
q_\phi^{A}(\widetilde z\mid x_0)
\,\|\, 
p_\psi^{A}(\widetilde z)
\right)
\right].
\end{aligned}
\end{equation}

The first term in \(\mathcal L_{\mathrm{VAMP\text{-}Diff}}\) is the \(x_0\)-prediction diffusion reconstruction objective conditioned on the full temporal latent \(z\) and the second term regularizes the pooled latent distribution toward the compact VampPrior. 
Hence, the decoder receives the full temporal latent needed for beat timing and morphology, whereas the VampPrior regularization is computed in a lower-dimensional latent space for tractability.
\end{theorem}

Theorem~\ref{thm:vampdiff_objective} separates two roles of the latent variable, i.e., the full latent \(z\in\mathbb{R}^{C_z\times T_z}\) conditions the diffusion decoder and preserves temporal information, while the pooled latent \(\widetilde z=Az\in\mathbb{R}^{C_z\times T_c}\) is used only in the VampPrior KL term, which makes mixture-prior regularization tractable without removing temporal information from the decoder. The annealing coefficient $\beta \geq 0$ follows a floor-then-ramp schedule discussed in Section~\ref{subsec:implementation}. Proof of Theorem~\ref{thm:vampdiff_objective} can be found in Appendix~\ref{pr:proof_thm}.

\subsection{Morphology-Preserving Loss Function}
\label{subsec:loss}

Theorem~\ref{thm:vampdiff_objective} gives the variational diffusion objective that couples the encoder, the compact VampPrior regularizer, and the conditional diffusion decoder. The objective function~\eqref{eq:vamp_diff} is sufficient to train the probabilistic model; however, PPG reconstruction and generation require more than local sample matching. It is worth noting that a signal can have low sample-wise error and still lose the sharp systolic upstroke, attenuate the dicrotic notch, distort harmonic content, or shrink the pulse amplitude. We therefore define auxiliary losses that penalize the specific waveform errors to preserve the PPG morphology.  We define the morphology-preserving loss function as
\begin{equation}
\label{eq:total_loss}
\begin{aligned}
\mathcal{L}_{\mathrm{MP\text{-}PPG}}(\theta,\phi,\psi)
&=
\mathcal{L}_{\mathrm{VAMP\text{-}Diff}}(\theta,\phi,\psi)
\\
&\hspace{-3.5em}
+ \lambda_{\mathrm{recon}}\mathcal{L}_{\mathrm{recon}}(x_0,\widehat{x}_0)
+ \lambda_{\mathrm{spec}}\mathcal{L}_{\mathrm{spec}}(x_0,\widehat{x}_0)
\\
&\hspace{-3.5em}
+ \lambda_{\mathrm{deriv}}\mathcal{L}_{\mathrm{deriv}}(x_0,\widehat{x}_0)
+ \lambda_{\mathrm{amp}}\mathcal{L}_{\mathrm{amp}}(x_0,\widehat{x}_0)
\\
&\hspace{-3.5em}
+ \lambda_{\mathrm{ptp}}\mathcal{L}_{\mathrm{ptp}}(x_0,\widehat{x}_0).
\end{aligned}
\end{equation}

where the coefficients
\(\lambda_{\mathrm{recon}},\lambda_{\mathrm{spec}},\lambda_{\mathrm{deriv}},
\lambda_{\mathrm{amp}},\lambda_{\mathrm{ptp}}\)
are fixed scalar weights that control the strength of the waveform penalties and \(\widehat{x}_0=f_\theta(x_t,t,z)\) is the clean signal prediction of the decoder.

The reconstruction loss compares the predicted clean signal with the target window at the sample level,
\begin{equation}
    \mathcal{L}_{\mathrm{recon}}(x_0,\widehat{x}_0)
    =
    \frac{1}{L}\sum_{i=1}^{L}
    \mathrm{SmoothL1}\!\left(
    \widehat{x}_0^{(i)},x_0^{(i)}
    \right).
    \label{eq:loss_recon}
\end{equation}
SmoothL1 is less sensitive to localized sample errors than squared loss but still penalizes systematic reconstruction bias. The spectral loss, on the other hand,  compares log-magnitude real-valued fast Fourier transform (FFT) coefficients,
\begin{equation}
\label{eq:spec_loss}
\begin{aligned}
\mathcal{L}_{\mathrm{spec}}(x_0,\widehat{x}_0)
&=
\frac{1}{F}\sum_{f=1}^{F}
\mathrm{SmoothL1}
\Big(
\log(1+|\widehat X_f|),
\\
&\hspace{4.5em}
\log(1+|X_f|)
\Big).
\end{aligned}
\end{equation}

where \(\mathcal F(\cdot)\), \(X_f=\mathcal F(x_0)[f]\), \(\widehat X_f=\mathcal F(\widehat{x}_0)[f]\), and \(|\cdot|\) are the real-valued FFT, the \(f\)-th frequency coefficient of the target signal, the coefficient of the reconstructed signal, and the complex magnitude, respectively. To penalize slope mismatch and preserve systolic upstrokes, we add the derivative loss,
\begin{equation}
\label{eq:loss_deriv}
\begin{aligned}
\mathcal{L}_{\mathrm{deriv}}(x_0,\widehat{x}_0)
&=
\frac{1}{L-1}\sum_{i=1}^{L-1}
\mathrm{SmoothL1}
\Big(
\widehat{x}_0^{(i+1)}-\widehat{x}_0^{(i)},
\\
&\hspace{4.5em}
x_0^{(i+1)}-x_0^{(i)}
\Big).
\end{aligned}
\end{equation}

We define \(   \mathcal{L}_{\mathrm{amp}}(x_0,\widehat{x}_0)\) to penalizes energy mismatch across the full segment as follows
\begin{equation}
    \mathcal{L}_{\mathrm{amp}}(x_0,\widehat{x}_0)
    =
    \left(
    \mathrm{std}(\widehat{x}_0)-\mathrm{std}(x_0)
    \right)^2,
    \label{eq:loss_amp}
\end{equation}
where \(\mathrm{std}(\cdot)\) is the standard deviation computed over the window.  The peak-to-peak loss accounts for the attenuation of pulse range
\begin{equation}
\label{eq:loss_ptp}
\mathcal{L}_{\mathrm{ptp}}(x_0,\widehat{x}_0)
=
\big[
\operatorname{ptp}(\widehat{x}_0)
-
\operatorname{ptp}(x_0)
\big]^2 ,
\end{equation}
where $\operatorname{ptp}(x)=\max(x)-\min(x)$. The standard-deviation and peak-to-peak losses are complementary, since two windows can have similar variance but different pulse heights, or similar peak-to-peak range but different energy distribution across the window.
Algorithms~\ref{alg:training} and~\ref{alg:generation} describe the training update with \(\mathcal L_{\mathrm{MP-PPG}}\) and unconditional generation from full-resolution VampPrior samples.

\begin{figure}[t]
\centering
\begin{minipage}[t]{0.48\textwidth}
\begin{algorithm}[H]
\caption{Training VAMP-Diff}
\label{alg:training}
\begin{algorithmic}[1]
\Repeat
    \State Sample $x_0 \sim p_{\mathrm{data}}$
    \State $(\mu_\phi(x_0),\sigma_\phi^2(x_0))=\mathrm{Encoder}_\phi(x_0)$
    \State Sample $\epsilon_z\sim\mathcal N(0,I)$
    \State $z=\mu_\phi(x_0)+\sigma_\phi(x_0)\odot\epsilon_z$
    \State $\widetilde z=A z$ \Comment{compact latent used only for KL}

    \For{$k=1,\ldots,K$}
        \State $(\mu_\phi(u_k),\sigma_\phi^2(u_k))=\mathrm{Encoder}_\phi(u_k)$
        \State $\widetilde\mu_k=A\mu_\phi(u_k)$
        \State $\widetilde\sigma_k^2=A\sigma_\phi^2(u_k)$
    \EndFor

    \State Define $q_\phi^A(\widetilde z\mid x_0)=A_{\#}q_\phi(z\mid x_0)$
    \State Define ${p_\psi^A(\widetilde z)=\dfrac{1}{K}\sum_{k=1}^{K}
    \mathcal N\!\left(\widetilde z;\widetilde\mu_k,
    \operatorname{diag}(\widetilde\sigma_k^2)\right)}$

    \State $\mathcal L_{\mathrm{KL}}
    =
    D_{\mathrm{KL}}\!\left(
    q_\phi^A(\widetilde z\mid x_0)
    \,\|\, 
    p_\psi^A(\widetilde z)
    \right)$

    \State Sample $t\sim\mathrm{Unif}\{1,\ldots,T\}$ and $\epsilon\sim\mathcal N(0,I)$
    \State $x_t=\sqrt{\bar\alpha_t}\,x_0+\sqrt{1-\bar\alpha_t}\,\epsilon$
    \State $\widehat x_0=f_\theta(x_t,t,z)$ \Comment{full temporal latent conditions decoder}

    \State $\mathcal L_{\mathrm{VAMP\text{-}Diff}}
    =
    \left\|x_0-\widehat x_0\right\|_2^2
    +
    \beta\,\mathcal L_{\mathrm{KL}}$

    \State Compute auxiliary losses
    $\mathcal L_{\mathrm{recon}}$, $\mathcal L_{\mathrm{spec}}$,
    $\mathcal L_{\mathrm{deriv}}$, $\mathcal L_{\mathrm{amp}}$,
    and $\mathcal L_{\mathrm{ptp}}$

    \State Compute $\mathcal L_{\mathrm{MP-PPG}}$ according to Eq.~\eqref{eq:total_loss}

    \State Update $(\theta,\phi,\psi)$ using $\nabla\mathcal L_{\mathrm{MP-PPG}}$
\Until{converged}
\end{algorithmic}
\end{algorithm}
\end{minipage}
\hfill
\begin{minipage}[t]{0.48\textwidth}
\begin{algorithm}[H]
\caption{Unconditional Generation}
\label{alg:generation}
\begin{algorithmic}[1]
    \State Sample component index $k\sim\mathrm{Unif}\{1,\ldots,K\}$
    \State $(\mu_\phi(u_k),\sigma_\phi^2(u_k))=\mathrm{Encoder}_\phi(u_k)$
    \State Sample $\epsilon_z\sim\mathcal N(0,I)$
    \State $z=\mu_\phi(u_k)+\sigma_\phi(u_k)\odot\epsilon_z$
    \Comment{full-resolution VampPrior sample}
    \State Sample $x_T\sim\mathcal N(0,I)$
    \For{$t=T,T-1,\ldots,1$}
        \State $\widehat x_0=f_\theta(x_t,t,z)$
        \State $\widehat\epsilon=
        \dfrac{x_t-\sqrt{\bar\alpha_t}\widehat x_0}
        {\sqrt{1-\bar\alpha_t}}$
        \State $x_{t-1}=
        \sqrt{\bar\alpha_{t-1}}\widehat x_0+
        \sqrt{1-\bar\alpha_{t-1}}\widehat\epsilon$
    \EndFor
    \State \Return $\widehat x_0$
\end{algorithmic}
\end{algorithm}
\end{minipage}
\end{figure}

\subsection{Architecture of VAMP-Diff}
\label{subsec:architecture_model}

VAMP-Diff combines latent inference, prior learning, and waveform synthesis in a single jointly trained architecture. Given an observed PPG window \(x_0\), the encoder outputs posterior parameters \(\mu_\phi(x_0)\) and \(\sigma_\phi^2(x_0)\) for the full temporal latent \(z\in\mathbb{R}^{C_z\times T_z}\), which conditions the diffusion decoder. The learnable pseudo-inputs \(\{u_k\}_{k=1}^{K}\) pass through the same encoder to form the VampPrior components used during training and generation. Figure~\ref{fig:architecture} shows the shared encoder structure and the distinction between the full latent \(z\) and the pooled latent \(\widetilde z=Az\). The KL comparison only utilizes the pooled latent \(\widetilde z=Az \in \mathbb{R}^{C_z\times T_c}\); the decoder receives the full temporal latent, which preserves temporal order through the diffusion updates. The conditional one-dimensional U-Net predicts \(\widehat{x}_0=f_\theta(x_t,t,z)\), and the decoder blocks receive \(z\) through multiscale FiLM modulation. This architecture keeps the KL computation tractable by applying prior matching only to the pooled latent. The diffusion decoder still receives the full temporal latent for reconstruction and generation, and the encoder, pseudo-inputs, and diffusion decoder are optimized jointly using the loss function~\eqref {eq:total_loss}.

\begin{figure*}[t]
    \centering
    \includegraphics[width=0.9\textwidth]{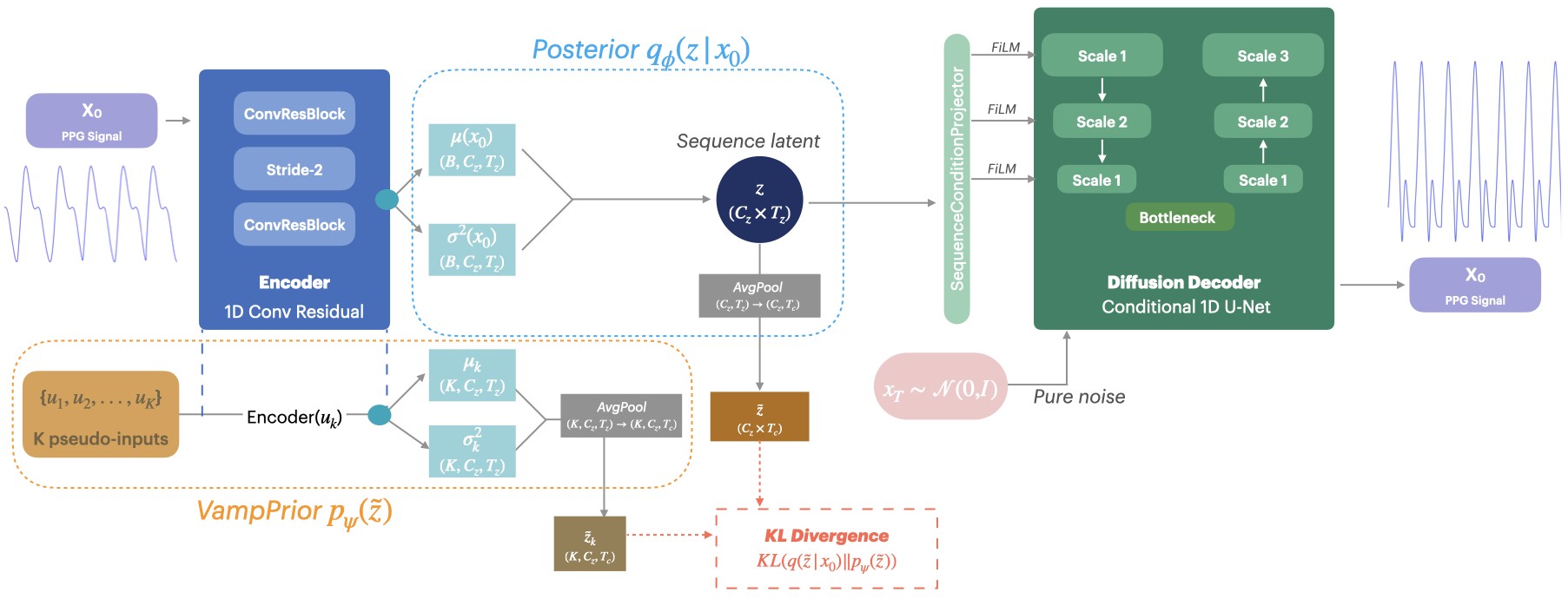}
    \caption{VAMP-Diff architecture showing the posterior path (top), VampPrior path (bottom), and diffusion decoder (right).}
    \label{fig:architecture}
\end{figure*}
\section{Downstream Tasks}
\label{sec:downstream}

We evaluate downstream behavior using reconstructed test signals to determine whether the learned latent representation preserves physiological content from the original PPG window. For each test signal \(x_0\), the encoder produces \(z=\mu_\phi(x_0)\), and the decoder reconstructs \(\widehat{x}_0\) through deterministic DDIM with \(\eta=0\), as described in Section~\ref{sec:decoder}. Both \(x_0\) and \(\widehat{x}_0\) are denormalized with the training-set statistics \(\mu_{\mathrm{train}}\) and \(\sigma_{\mathrm{train}}\). This process evaluates the encoder and decoder jointly, since the latent must keep a physiologically relevant structure and the diffusion decoder must recover that structure through the denoising process. We consider three test-split evaluations, heart-rate estimation, respiratory-rate estimator consistency, and reconstruction-based waveform corruption detection.

\subsection{Heart-Rate Estimation}
\label{subsec_hr_estimation}

Heart rate (HR) estimation is the first physiological test for PPG, as cardiac rhythm is the dominant periodic component of clean PPG, and accurate recovery depends primarily on inter-beat timing rather than pointwise waveform similarity~\cite{elgendi2012ppg, charlton2022wearable}. For each test window, both \(x_0\) and \(\widehat{x}_0\) are bandpass-filtered to \(0.7\)--\(3.0\) Hz before cardiac feature extraction. Systolic peaks are detected with minimum inter-peak distance \(0.35\) s, prominence \(0.1\sigma\) of the filtered signal, and height at the \(60\)th percentile. Per-window HR is computed as \(60/\bar{\mathrm{IBI}}\), where \(\bar{\mathrm{IBI}}\) is the mean inter-beat interval. Signal-level reconstruction is evaluated using MAE, RMSE, and Pearson correlation. Cardiac timing fidelity is assessed through HR absolute error and IBI mean absolute error on the VAMP-Diff model.

\subsection{Respiratory Rate Estimation}
\label{subsec_rr_estimation}

Respiratory rate (RR) is a more demanding physiological test, since respiration affects PPG through baseline variation, amplitude modulation, and frequency modulation in the \(0.1\)--\(0.5\) Hz band~\cite{charlton2022wearable}. We use a frozen-estimator procedure so that RR evaluation stays independent of VAMP-Diff training. A one-dimensional ResNet is trained once on real PPG to regress RR from CO\(_2\)-derived labels, and the generative model is not used during estimator training. The estimator uses a stride-2 stem with kernel size \(11\) and \(32\) channels, followed by four residual stages with channel widths \(\{32,64,128,128\}\), dilation rates \(\{1,2,4,8\}\), global average pooling, and a two-layer MLP regression head. Ground-truth RR is obtained from the synchronous EtCO\(_2\) channel in CapnoBase~\cite{karlen2021capnobase} by peak detection on a five-sample-smoothed signal, with minimum inter-peak distance \(1\) s and prominence \(0.05\). Recordings outside \([6,35]\) breaths per minute are excluded. We apply the frozen estimator to real and reconstructed PPG windows, and use \(\mathrm{MAE}_{\mathrm{recon}}-\mathrm{MAE}_{\mathrm{real}}\) to measure how much RR-relevant information changes after reconstruction.

\section{Empirical Results}
\label{sec:empirical_results}


The proposed VAMP-Diff model is evaluated along two axes: generative quality and reconstruction fidelity, validated through model comparison and downstream physiological inference tasks. The former tests whether the model produces PPG signals that are statistically and morphologically consistent with real data; the latter tests whether the encoded latent $z$ carries sufficient physiological information to recover clinically meaningful quantities through the decode pipeline, as measured by the downstream tasks of Section~\ref{sec:downstream}.

\subsection{Dataset and Preprocessing}
\label{subsec:dataset_preproc}
Experiments use the CapnoBase dataset~\cite{karlen2021capnobase}, a freely available ICU recording containing PPG and capnography signals from 42 patients, each recorded for approximately 8 minutes at 300 Hz. The PPG channel (\texttt{pleth\_y}) provides the input signal. Respiratory-rate ground truth is derived from the synchronous EtCO$_2$ channel (\texttt{CO$_2$\_y}) as described in Section~\ref{subsec_rr_estimation}.

Each recording is segmented into overlapping windows of $3{,}072$ samples ($10.24$ s). A bandpass filter ($0.7$--$3.0$ Hz) is applied transiently for peak-based quality checking only; windows with fewer than two detected systolic peaks are discarded. The stored windows retain the full signal spectrum, including the respiratory modulation band ($0.1$--$0.5$~Hz), which is necessary for downstream respiratory rate estimation. Windows are normalized using the training-set mean $\mu_{\text{train}}$ and the standard deviation $\sigma_\text{train}$, retained for denormalization at evaluation.
The dataset yields $6{,}815$ training, $1{,}410$ validation, and $1{,}645$ test windows. Recordings are split at the patient level, ensuring that no patient appears in more than one split and preventing data leakage between training and evaluation. 

\subsection{Implementation Details}
\label{subsec:implementation}

Architecture details are described in Section~\ref{sec:hybridmodel}. The encoder and decoder are trained jointly after a 20-epoch freeze phase that allows the decoder and pseudo-inputs to stabilize before the encoder is updated. Separate learning rates are used for each component because they play different roles: the pseudo-inputs require a higher rate to track the changing posterior, while the encoder uses a lower rate to avoid destabilizing the latent space. The KL weight $\beta$ is annealed rather than fixed to prevent posterior collapse early in training. Auxiliary loss weights $\lambda_{\mathrm{recon}},\lambda_{\mathrm{spec}}, \lambda_{\mathrm{deriv}},\lambda_{\mathrm{amp}},\lambda_{\mathrm{ptp}}$ were selected empirically on validation reconstructions to balance pointwise reconstruction quality, heart rate preservation, amplitude fidelity, and absence of high-frequency artifacts. All hyperparameters are summarized in Table~\ref{tab:hyperparams}.

\begin{table*}[t]
\footnotesize
\centering
\caption{Training hyperparameters.}
\begin{tabular}{lll}
\toprule
Component & Parameter & Value \\
\midrule
\multirow{4}{*}{Architecture} 
    & Encoder downsampling & $L=3072 \rightarrow T_z=768$ \\
    & Latent channels & $C_z=256$ \\
    & Diffusion timesteps & $T=100$ with linear noise schedule \\
    & DDIM steps / $\eta$ & $50$ / $0$ \\
\midrule
\multirow{2}{*}{VampPrior}
    & Pseudo-inputs $K$ & $100$ \\
    & Initialization & Stratified real windows \\
\midrule
\multirow{5}{*}{Optimizer}
    & Algorithm & AdamW \\
    & Decoder LR & $2\times10^{-5}$ \\
    & Encoder LR & $5\times10^{-6}$ \\
    & Pseudo-input LR & $2\times10^{-3}$ \\
    & Weight decay / grad clip & $10^{-4}$ / $1.0$ \\
\midrule
\multirow{2}{*}{Training}
    & Epochs / batch size & $200$ / $32$ \\
    & Encoder freeze & Epochs $1$--$20$ \\
\midrule
\multirow{3}{*}{KL annealing}
    & Frozen phase & $\beta = 0$ \\
    & Floor phase (ep.\ 21--50) & $\beta = 10^{-8}$ \\
    & Ramp phase (ep.\ 51--130) & $\beta \rightarrow 5\times10^{-8}$ \\
\midrule
\multirow{1}{*}{Loss weights}
    & $\lambda_{\mathrm{diff}}, \lambda_{\mathrm{recon}}, 
    \lambda_{\mathrm{spec}}, \lambda_{\mathrm{deriv}}, 
    \lambda_{\mathrm{amp}}, \lambda_{\mathrm{ptp}}$ & 
    $1.0,\ 5.0,\ 0.1,\ 0.1,\ 2.0,\ 1.0$ \\
\bottomrule
\end{tabular}
\label{tab:hyperparams}
\end{table*}

\subsection{ Empirical Result 1: Model Comparison}
\label{subsec:model_comparisons}
To highlight the benefits of the proposed design, we compare our model to a Vanilla VAE with a similar encoder architecture and a normal Gaussian, and the proposed VAE-Diff model with a normal Gaussian prior.

\subsubsection{Reconstruction Quality}
\label{subsubsec:reconstruction_quality}
In Table~\ref{tab:recon_metrics}, compares reconstruction quality across three models. VAMP-Diff achieves comparable pointwise reconstruction to VAE-Diff while reducing HR absolute error. The Vanilla VAE produces the lowest pointwise error but the highest HR reconstruction error ($3.5$ bpm), since its decoder averages over the conditional distribution and blurs peak timing. HR reconstruction error decreases sharply when the decoder is swapped for the diffusion model (a 77\% reduction, from $3.5\rightarrow 0.8$ bpm), meaning that the diffusion decoder preserves temporal periodicity much better because it learns the full conditional distribution, not just the mean. Adding the VampPrior further reduces HR error ($0.8 \rightarrow 0.556$ bpm). This supports the claim that the VampPrior regularizes the posterior toward a more structured prior, encouraging the encoder to represent temporal periodicity more accurately. Figure~\ref{fig:reconstruct_samples} shows representative reconstructions from all three models on the same test window. 

Unconditional generation reveals a more fundamental difference between models. When sampling from $\mathcal{N}(0, I)$, the Gaussian prior does not match the aggregate posterior learned by the encoder, so generation-time samples fall outside the latent regions seen during training. The VampPrior addresses this by adapting the prior to the encoder's aggregate posterior, enabling VAMP-Diff to produce realistic PPG unconditionally shown in Figure~\ref{fig:gen_comparison}.

\begin{table}[t]
\footnotesize
\centering
\caption{Reconstruction metrics on the test set (mean $\pm$ std)}
\resizebox{\columnwidth}{!}{%
\begin{tabular}{lcccc}
\toprule
Model & MAE & RMSE & Corr & HR Abs Err (bpm) \\
\midrule
Vanilla VAE &
    0.045$\pm$0.018 &
    0.055$\pm$0.021 &
    0.9997$\pm$0.0004 &
    3.5$\pm$2.0  \\
VAE-Diff (no VampPrior) &
    0.140$\pm$0.030 & 
    0.160$\pm$0.032 & 
    0.9992$\pm$0.0012 & 
    0.8$\pm$2.0 \\
VAMP-Diff  &
    0.151$\pm$0.031 & 
    0.170$\pm$0.033 & 
    0.9991$\pm$0.0013 & 
    0.556$\pm$1.904 \\
\bottomrule
\end{tabular}}
\label{tab:recon_metrics}
\end{table}

\begin{figure*}[t]
    \centering

    \subfloat[Vanilla VAE]{
        \includegraphics[width=0.30\textwidth]{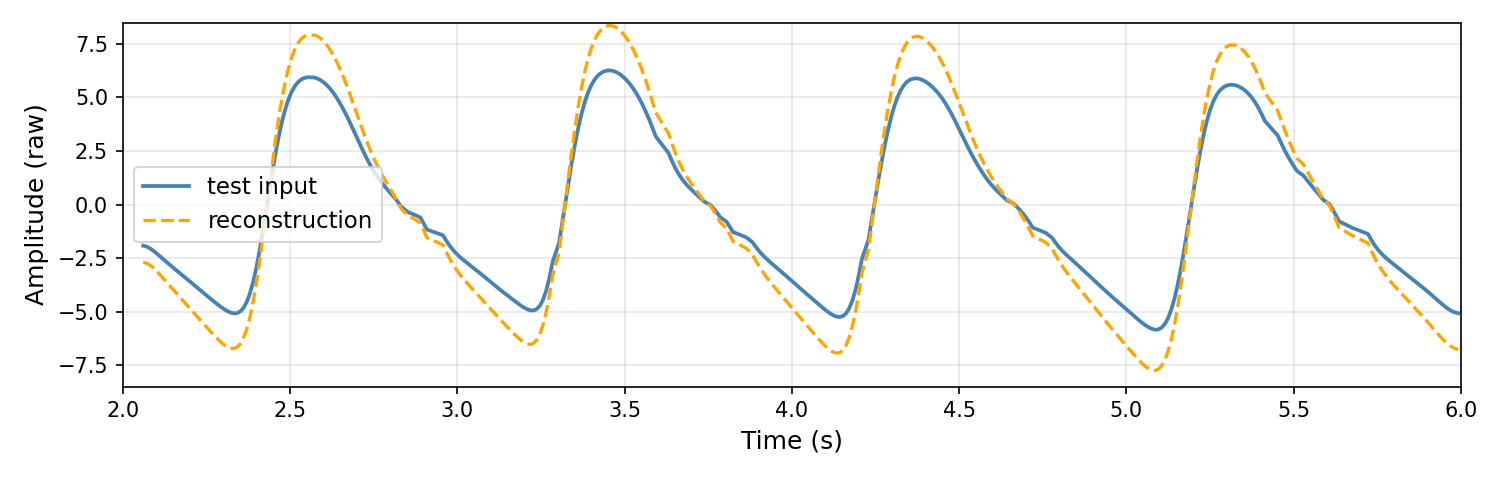}
        \label{subfig:vanilla_vae}
    }
    \hfil
    \subfloat[VAE-Diff]{
        \includegraphics[width=0.30\textwidth]{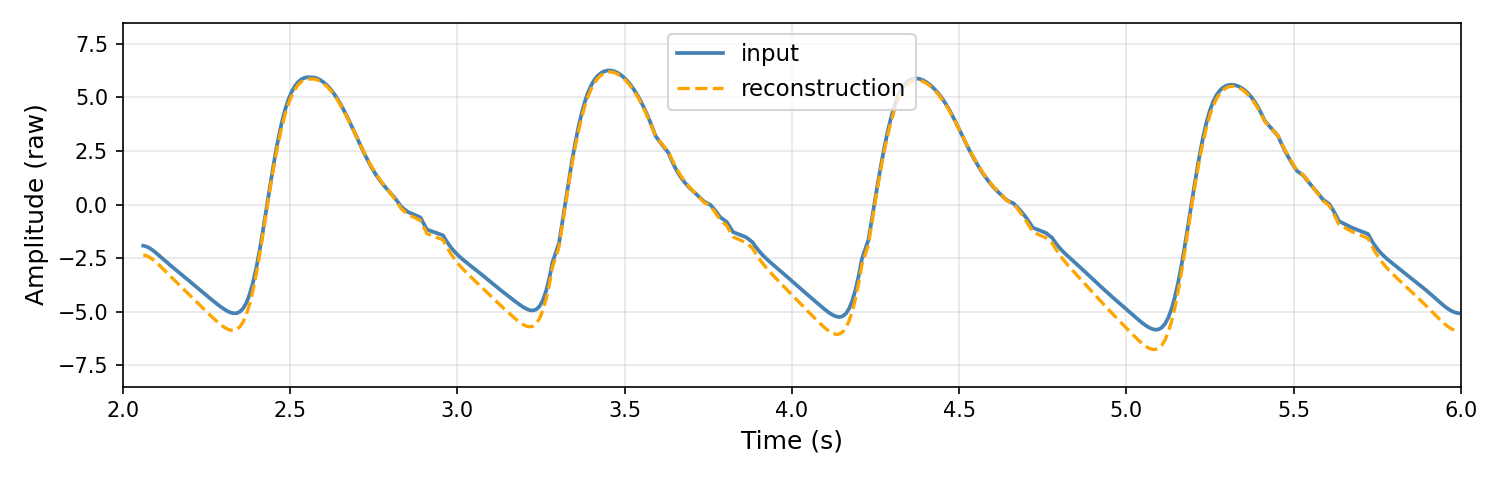}
        \label{subfig:vae_diff}
    }
    \hfil
    \subfloat[VAMP-Diff]{
        \includegraphics[width=0.30\textwidth]{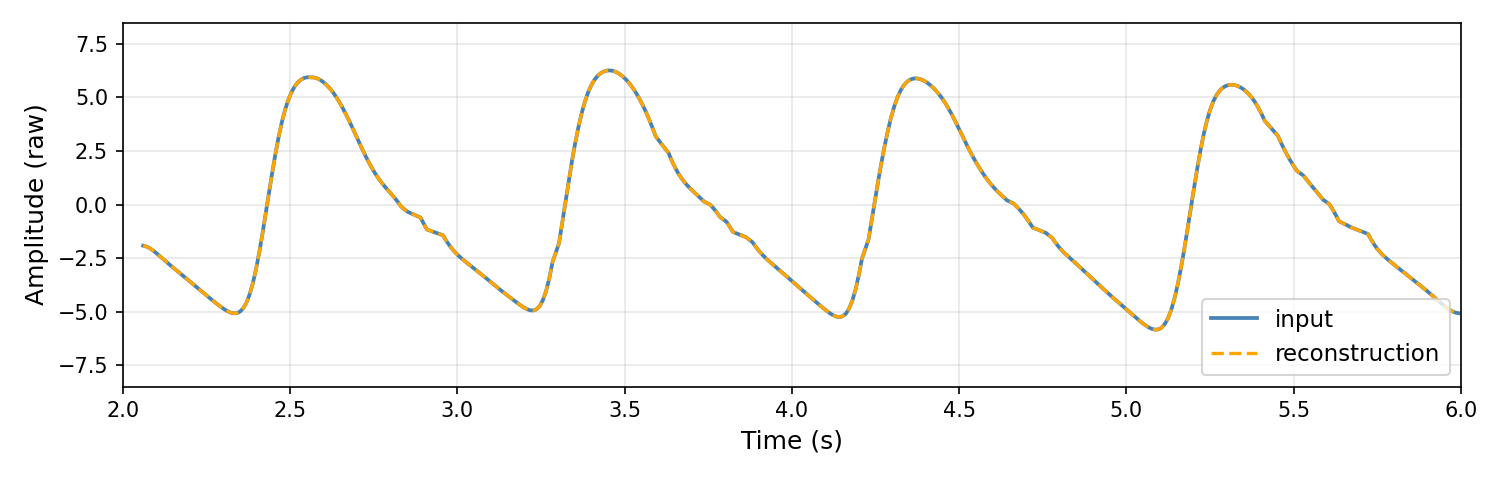}
        \label{subfig:vamp}
    }

    \caption{Reconstruction of the same test PPG window by three models. 
    (a) Vanilla VAE: good pointwise reconstruction but morphologically smoothed. 
    (b) VAE-Diff: sharper morphology due to the diffusion decoder, but still limited by the Gaussian prior. 
    (c) VAMP-Diff: best morphological fidelity with the lowest HR error, as the VampPrior enables a more structured latent space.}
    \label{fig:reconstruct_samples}
\end{figure*}

\begin{figure}[t]
    \centering

    \subfloat[Real PPG\label{subfig:gen_real}]{
        \includegraphics[width=0.95\linewidth]{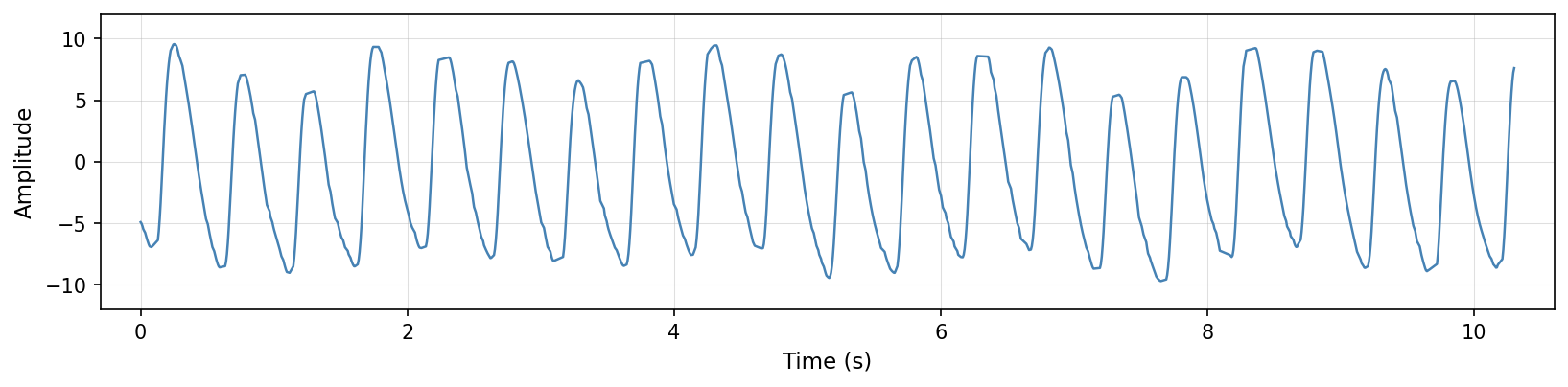}
    }\\[0.1em]

    \subfloat[Vanilla VAE\label{subfig:gen_vae}]{
        \includegraphics[width=0.95\linewidth]{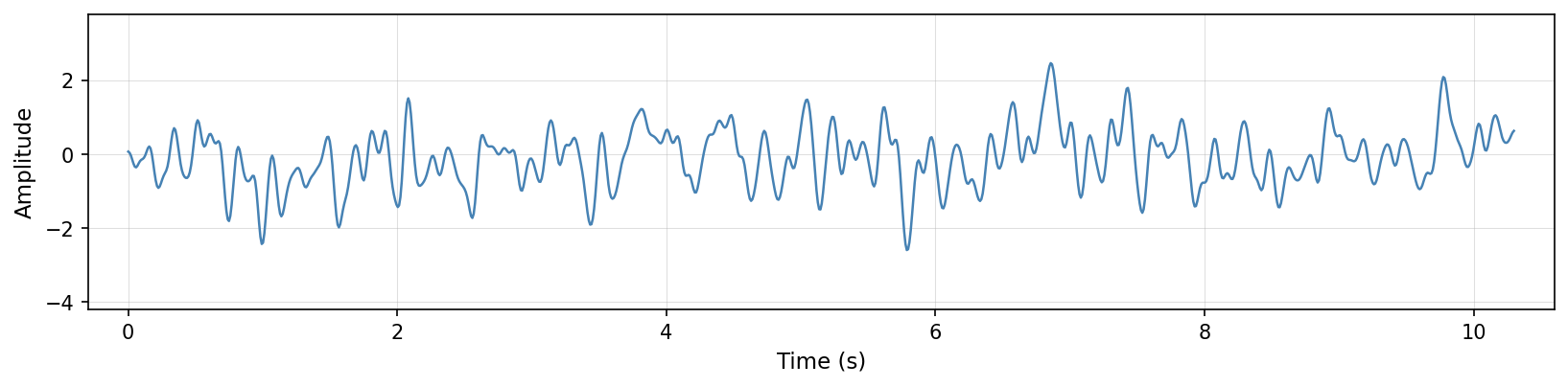}
    }\\[0.1em]

    \subfloat[VAE-Diff\label{subfig:gen_vaediff}]{
        \includegraphics[width=0.95\linewidth]{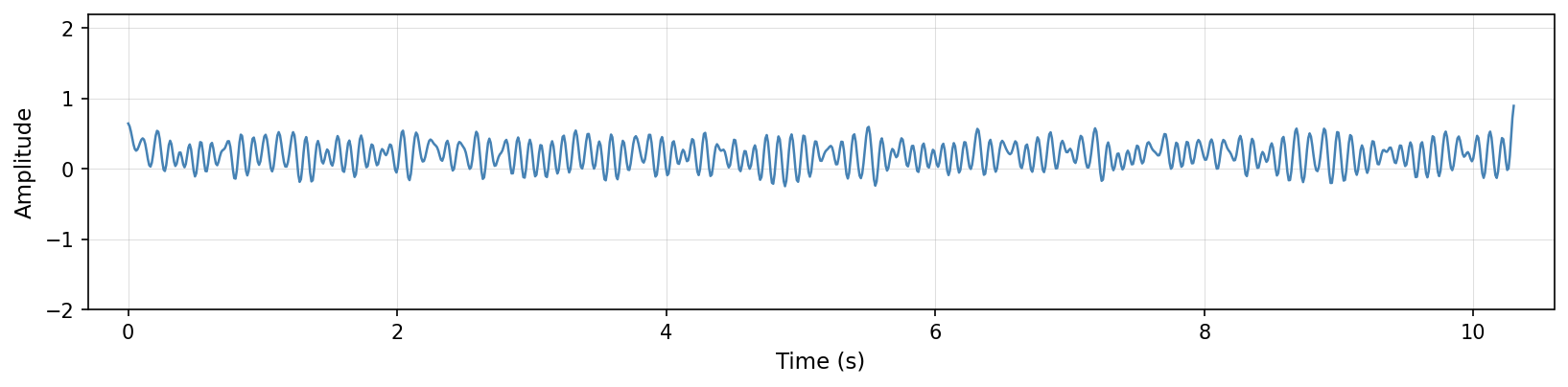}
    }\\[0.1em]

    \subfloat[VAMP-Diff\label{subfig:gen_vampdiff}]{
        \includegraphics[width=0.95\linewidth]{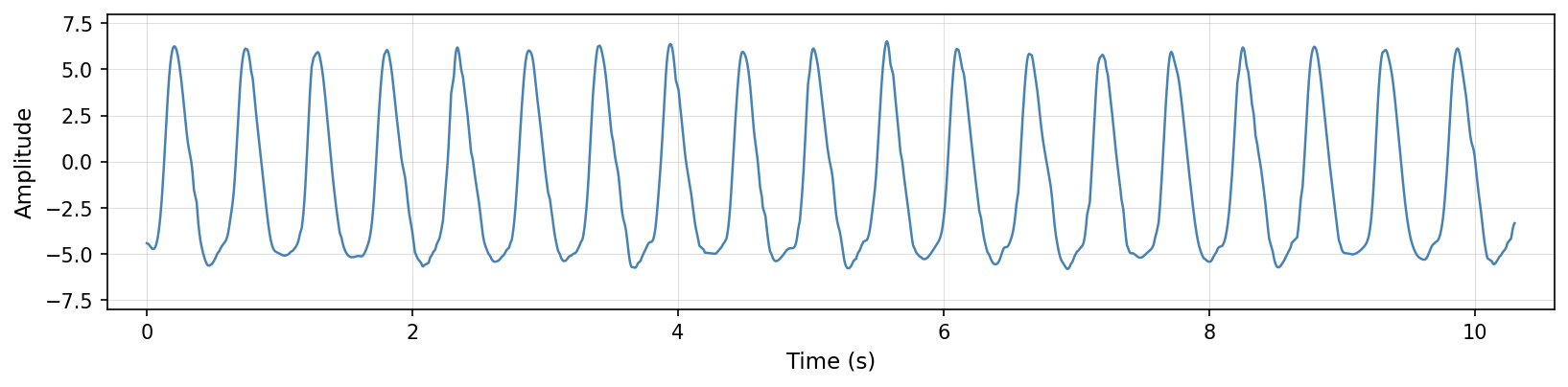}
    }

    \caption{Unconditional generation samples compared to a real PPG window. 
    (a) Real PPG: clear systolic peaks with realistic morphology and amplitude modulation. 
    (b) Vanilla VAE: no periodic structure, indicating generation failure without a structured latent prior. 
    (c) VAE-Diff (Gaussian prior): periodic structure present but amplitude collapses near zero, consistent with prior-posterior mismatch at generation time. 
    (d) VAMP-Diff: clean periodic waveform with morphology closest to the real signal.}
    \label{fig:gen_comparison}
\end{figure}

\subsubsection{Ablation Study} 
Table~\ref{tab:ablation} isolates the contribution of each design choice. Removing KL regularization yields the best pointwise reconstruction but collapses the latent space, making generation impossible. A Gaussian prior causes generation failure due to prior-posterior mismatch. Replacing the temporal sequence latent with a global pooled vector preserves pointwise reconstruction quality but fails at generation, confirming that temporal structure in $z$ is necessary for HR-faithful synthesis. Removing auxiliary losses causes generation failure, indicating that morphological supervision is necessary for the decoder to produce physiologically valid signals. Random pseudo-input initialization degrades both reconstruction and generation.  Stratified initialization restores reconstruction quality but leaves a $13.2$ bpm HR gap because the compressed prior loses temporal periodicity. The full-resolution prior reduces the HR gap to $1.2$ bpm, confirming that the full-resolution prior is the critical component for generation fidelity.

\begin{table}[t]
\footnotesize
\centering
\caption{Ablation study results. Corr is measured on the test set reconstruction. HR Gap is $|\mu_{\text{HR,gen}} - \mu_{\text{HR,real}}|$ and Peak Det. is the fraction of generated signals with at least two detected peaks.}
\begin{tabular}{lccc}
\toprule
Configuration & Corr & HR Gap (bpm) & Peak Det. \\
\midrule
No KL regularization & 
    0.9998 & N/A & $<0.05$ \\
Gaussian prior & 
    0.9988 & N/A & $<0.05$ \\
Pooled latent (global $z$) & 
    0.9988 &  N/A & $<0.05$ \\
No auxiliary losses & 
    0.9998 &  N/A & $<0.05$ \\
VampPrior, random init & 
    0.9820 & 12.0 & $>0.90$ \\
VampPrior, stratified & 
    0.9991 & 13.2 & $>0.90$ \\
\textbf{VampPrior, full-res (proposed)} & 
    \textbf{0.9991} & \textbf{1.2} & 
    \textbf{0.97} \\
\bottomrule
\end{tabular}\\[4pt]
\footnotesize
\label{tab:ablation}
\end{table}

\subsection{Empirical Result 2: Downstream Tasks}
\label{subsec:downstream_tasks}
This section evaluates whether the generated signals maintain the physiological properties needed for the estimation of heart rate and respiratory rate.
Before evaluating downstream tasks, we first verify that the decoder actively uses the latent $z$. Replacing $\mu_\phi(x_0)$ with $z_{\text{rand}} \sim \mathcal{N}(0, I)$ while fixing $x_T$ shifts the decoder output by approximately 21\% of the signal range (Table~\ref{tab:z_sensitivity}), confirming that $z$ actively guides the denoising process rather than being bypassed. However, sensitivity to $z$ alone does not guarantee that $z$ encodes clinically meaningful content, only that the decoder uses the information that it carries. The following evaluations test whether that information corresponds to heart rate, respiratory rate, and waveform morphology---quantities required for clinical use. The sensitivity ratio is defined as
{\setlength{\abovedisplayskip}{4pt}
\setlength{\belowdisplayskip}{4pt}
\begin{equation}
\label{eq:z_sensitivity}
\rho =
\frac{
\mathrm{MAD}\!\left(
\widehat{x}_0(\mu_\phi(x_0)),
\widehat{x}_0(z_{\mathrm{rand}})
\right)
}{
\mathrm{range}\!\left(
\widehat{x}_0(\mu_\phi(x_0))
\right)
}.
\end{equation}
}

\begin{table}[t]
\footnotesize
\centering
\caption{Latent sensitivity. Non-zero ratio 
confirms the decoder output changes when $z$ changes.}
\begin{tabular}{lccc}
\toprule
Sample & MAD & Output range & Ratio $\rho$ \\
\midrule
0 & 1.177 & 5.419 & 0.217 \\
1 & 1.159 & 5.419 & 0.214 \\
2 & 1.157 & 5.406 & 0.214 \\
\bottomrule
\end{tabular}
\label{tab:z_sensitivity}
\end{table}

\subsubsection{Heart Rate Estimation}
Heart rate is estimated from inter-beat intervals detected in the bandpass-filtered signal ($0.7$-$3.0$ Hz) using a peak detector with minimum inter-peak distance $0.35$~s, prominence threshold $0.1\sigma$, and height at the 60th percentile. Per-window HR is computed as $60 / \bar{\mathrm{IBI}}$, where $\bar{\mathrm{IBI}}$ is the mean inter-beat interval.

Table~\ref{tab:hr_results} reports the comparison across 5{,}000 generated and 1{,}645 real test signals. Mean HR ($90.0$ vs. $90.7$ bpm) and peak count ($15.2$ vs $15.5$) match real data closely. However, the generated distribution shows a narrower spread in HR variance ($18.3$ vs. $21.4$ bpm), PTP amplitude ($1.88$ vs $5.14$), and signal standard deviation ($0.28$ vs $1.63$), indicating that while first-order statistics are well preserved, amplitude and rate diversity are more constrained than in real data. The mean pairwise distance $3.99 \pm 0.65$ confirms that generation does not collapse to a single mode. Distributional differences are confirmed by KS tests (Table~\ref{tab:ks_tests}), with all three metrics, with the largest discrepancy is in signal standard deviation (KS$=0.486$), consistent with the amplitude under-dispersion observed in Table~\ref{tab:hr_results}. The generated distribution underrepresents low heart rates (below $75$ bpm) and slightly overrepresents high heart rates (above $115$ bpm), reflecting residual distributional differences despite the VampPrior spanning the training HR range. 

\begin{table}[t]
\footnotesize
\centering
\caption{ Physiological statistics of 5{,}000 generated vs 1{,}645 real PPG signals.}
\begin{tabular}{lcc}
\toprule
Metric & Generated & Real \\
\midrule
HR (bpm) & $90.0 \pm 18.0$ & $90.7 \pm 21.4$ \\
IBI (ms) & $647.2 \pm 158.2$ & $656.3 \pm 183.2$ \\
Peak Count & $15.2 \pm 2.9$ & $15.5 \pm 3.7$ \\
PTP Amplitude & $14.13 \pm 1.88$ & $13.65 \pm 5.14$ \\
Signal Std & $3.73 \pm 0.28$ & $3.79 \pm 1.63$ \\
Mean pairwise dist. & $3.99 \pm 0.65$ & — \\
\bottomrule
\end{tabular}
\label{tab:hr_results}
\end{table}

\begin{table}[t]
\footnotesize
\centering
\caption{Kolmogorov--Smirnov tests between generated ($N=5{,}000$) and real ($N=1{,}645$) signal distributions. The KS statistic is the maximum absolute difference between empirical cumulative distribution functions, ranging from $0$ (identical) to $1$ (no overlap).}
\begin{tabular}{lcc}
\toprule
Metric & KS statistic \\
\midrule
Heart rate (bpm) & 0.181 \\
PTP amplitude    & 0.232 \\
Signal std       & 0.486 \\
\bottomrule
\end{tabular}
\label{tab:ks_tests}
\end{table}

\subsubsection{Respiratory Rate Estimation}
A pretrained CNN-based RR estimator is evaluated in two scenarios to determine if the proposed model preserves respiratory rate information.

\paragraph{Reconstruction}
Table~\ref{tab:rr_estimator} reports the estimator's performance on real held-out PPG signals against EtCO$_2$ ground truth, establishing a baseline before applying it to reconstructed signals. The estimator achieves moderate accuracy on real signals (MAE$=6.63$ bpm, $r=0.669$), which limits the strength of conclusions drawn from its predictions.

\begin{table}[t]
\footnotesize
\centering
\caption{RR estimator performance on real PPG signals against EtCO$_2$ ground truth.}
\begin{tabular}{lc}
\toprule
Metric & Value \\
\midrule
MAE (bpm) & 6.63 \\
RMSE (bpm) & 7.18 \\
Pearson $r$ & 0.669 \\
\bottomrule
\end{tabular}
\label{tab:rr_estimator}
\end{table}

\paragraph{Generation}
Table~\ref{tab:rr_result} compares predictions over 5{,}000 generated and 1{,}645 real signals to evaluate distributional agreement at scale. The mean predicted RR differs by only $0.3$ bpm ($15.5$ vs $15.8$ bpm), with similar distributional spread ($2.3$vs. $3.3$ bpm standard deviation).

\begin{table}[t]
\footnotesize
\centering
\caption{Predicted RR distribution across 5{,}000 
        generated and 1{,}645 real signals (bpm).}
\begin{tabular}{lcccc}
\toprule
        Metric & Generated & Real\\
        \midrule
        Pred RR mean$\pm$std (bpm) & 
            $15.5 \pm 2.3$ & $15.8 \pm 3.3$  \\
        Pred RR range (bpm)  & $11.4 - 23.1$ & $10.2 - 21.3$\\
\bottomrule
\end{tabular}
\label{tab:rr_result}
\end{table}

The predicted RR distribution has a similar mean to the real signals (difference in $0.3$ bpm), but generated signals produce a smoother unimodal distribution compared to the bimodal real distribution, consistent with the model not explicitly encoding respiratory rate. Since ground-truth RR is not available for generated signals, these predictions reflect the estimator's behavior on generated PPG rather than the true respiratory rate.

\section{Conclusions}
\label{sec:conclusions}

This paper presented VAMP-Diff, a jointly trained variational diffusion model for PPG signal modeling that combines a temporal encoder, a conditional one-dimensional diffusion decoder, and a VampPrior regularizer on a compact pooled latent. The model passes the full temporal latent to the decoder during reconstruction and generation, and hence beat timing, pulse morphology, amplitude structure, and respiratory modulation stay available to the denoising process, and the compact VampPrior KL makes prior matching tractable without forcing generation through a fixed Gaussian prior. Experiments on CapnoBase show that VAMP-Diff generates physiologically realistic PPG waveforms, improves heart-rate preservation over Gaussian-prior baselines, maintains respiratory-rate estimator consistency, and shows reconstruction-error sensitivity to waveform corruptions.

\bibliographystyle{ieeetr}
\bibliography{references}

\clearpage
\appendices
\section{Proof of Theorem~\ref{thm:vampdiff_objective}}
\label{pr:proof_thm}
\begin{proof}
Let \(q_\phi^A(\widetilde z\mid x_0)=A_{\#}q_\phi(z\mid x_0)\) and
\(p_\psi^A(\widetilde z)=\frac{1}{K}\sum_{k=1}^K q_\phi^A(\widetilde z\mid u_k)\). Given the variational identity on the compact latent space, we have
\begin{equation}
\begin{aligned}
\log p_{\theta,\psi}(x_0)
&=
\log\int
p_\theta(x_0\mid z)\,
p_\psi^A(\widetilde z)\,
d\widetilde z
\\
&\ge
\mathbb E_{q_\phi(z\mid x_0)}
\big[
\log p_\theta(x_0\mid z)
\big]
\\
&\quad
-
D_{\mathrm{KL}}\!\left(
q_\phi^A(\widetilde z\mid x_0)
\,\|\,
p_\psi^A(\widetilde z)
\right).
\end{aligned}
\end{equation}
where \(\widetilde z=Az\). The KL term is well-defined since both distributions live on the same compact latent space \(\mathbb R^{C_z\times T_c}\) and the conditional likelihood keeps the full latent \(z\in\mathbb R^{C_z\times T_z}\). Given the diffusion forward process~\ref{eq:diff_forward} and the reverse model 

\begin{equation}
\label{eq:LB}
\begin{aligned}
\log p_\theta(x_0\mid z)
&=
\log\int p_\theta(x_{0:T}\mid z)\,dx_{1:T}
\\
&=
\log\int q(x_{1:T}\mid x_0)
\frac{p_\theta(x_{0:T}\mid z)}
{q(x_{1:T}\mid x_0)}
\,dx_{1:T}
\\
&\geq
\int q(x_{1:T}\mid x_0)
\log
\frac{p_\theta(x_{0:T}\mid z)}
{q(x_{1:T}\mid x_0)}
\,dx_{1:T}.
\end{aligned}
\end{equation}
Using the Markov factorizations and the exact posterior
\(q(x_{t-1}\mid x_t,x_0)\), the lower bound~\ref{eq:LB} decomposes into

\begin{equation}
\begin{aligned}
\mathcal L_{\mathrm{diff}}(x_0,z)
&=
- \underbrace{
D_{\mathrm{KL}}\!\left(
q(x_T\mid x_0)\,\|\,p(x_T)
\right)
}_{\text{independent of } \theta}
\\
&\quad
+
\mathbb E_q\!\left[
\log p_\theta(x_0\mid x_1,z)
\right]
\\
&\quad
-
\sum_{t=2}^{T}
\mathbb E_q
\left[
\begin{aligned}
&D_{\mathrm{KL}}\!\left(
q(x_{t-1}\mid x_t,x_0)
\,\|\, \right.
\\
&\left.
p_\theta(x_{t-1}\mid x_t,z)
\right)
\end{aligned}
\right].
\end{aligned}
\end{equation}
Moreover, under the standard diffusion noise schedule, we have \(\bar\alpha_T\approx0\) and \(q(x_T\mid x_0)\approx\mathcal N(0,I)=p(x_T),\)
which makes the terminal KL term negligible. for \(t\geq 2\), we assume

\begin{equation}
\begin{aligned}
q(x_{t-1}\mid x_t,x_0)
&=
\mathcal N\!\left(
\widetilde\mu_t(x_t,x_0),
\widetilde\beta_t I
\right),
\\
p_\theta(x_{t-1}\mid x_t,z)
&=
\mathcal N\!\left(
\mu_\theta(x_t,t,z),
\sigma_t^2 I
\right).
\end{aligned}
\end{equation}
We can compute the KL term as follows
\begin{equation}
\begin{aligned}
&D_{\mathrm{KL}}\!\left(
q(x_{t-1}\mid x_t,x_0)
\,\|\,
p_\theta(x_{t-1}\mid x_t,z)
\right)
\\
&\quad =
\frac{1}{2\sigma_t^2}
\left\|
\widetilde\mu_t(x_t,x_0)
-
\mu_\theta(x_t,t,z)
\right\|_2^2
+
C_t .
\end{aligned}
\end{equation}
where \(C_t\) is independent of \(\theta\). Following the assumptions made in \cite{ho2020denoising}, the exact DDPM posterior mean under \(x_0\)-prediction parameterization and the reverse posterior mean with the network prediction
\(\widehat x_0=f_\theta(x_t,t,z)\) are given by 
\begin{equation*}
\begin{aligned}
\widetilde\mu_t(x_t,x_0)
&=
\frac{\sqrt{\bar\alpha_{t-1}}\beta_t}
{1-\bar\alpha_t}x_0
+
\frac{\sqrt{\alpha_t}(1-\bar\alpha_{t-1})}
{1-\bar\alpha_t}x_t,
\\
\mu_\theta(x_t,t,z)
&=
\frac{\sqrt{\bar\alpha_{t-1}}\beta_t}
{1-\bar\alpha_t}
f_\theta(x_t,t,z)
+
\frac{\sqrt{\alpha_t}(1-\bar\alpha_{t-1})}
{1-\bar\alpha_t}x_t .
\end{aligned}
\end{equation*}
and therefore
\begin{equation}
\begin{aligned}
&D_{\mathrm{KL}}\!\left(
q(x_{t-1}\mid x_t,x_0)
\,\|\,
p_\theta(x_{t-1}\mid x_t,z)
\right)
\\
&\quad =
w_t
\left\|
x_0-f_\theta(x_t,t,z)
\right\|_2^2
+
C_t,
\\
&w_t
=
\frac{\bar\alpha_{t-1}\beta_t^2}
{2\sigma_t^2(1-\bar\alpha_t)^2}.
\end{aligned}
\end{equation}
Similar to the simplified DDPM objective \cite{ho2020denoising}, the positive scalar \(w_t\) is dropped, which preserves the same minimizer up to timestep reweighting. Given forward process~\ref{eq:diff_forward}, the diffusion negative lower-bound becomes
\begin{equation}
  \mathbb E_{t,\epsilon,z}
\left[
\left\|
x_0-
f_\theta
\!\left(
\sqrt{\bar\alpha_t}x_0+\sqrt{1-\bar\alpha_t}\epsilon,
t,
z
\right)
\right\|_2^2
\right].  
\end{equation}
Adding the compact VampPrior KL with annealing weight \(\beta\) gives the loss function~\ref{eq:vamp_diff} up to constants independent of \((\theta,\phi,\psi)\) and the standard positive diffusion timestep weights. The decoder therefore uses the full sequence latent \(z\), whereas prior regularization is applied only to the pooled latent \(\widetilde z\).
\end{proof}

\section{Waveform Corruption Detection}
\label{subsec_rr_anomaly}

CapnoBase does not provide annotated abnormal PPG events. To check whether reconstruction error is sensitive to abnormal waveform structure, we use synthetic waveform corruptions as a sanity check with reconstruction error, and not as a standalone anomaly benchmark. Clean test windows are corrupted with additive Gaussian noise, low-frequency baseline wander, amplitude clipping, and partial flatline. These corruptions represent broadband noise, slow drift, dynamic-range saturation, and signal dropout. The trained model is applied to clean and corrupted windows without anomaly-specific training, and per-window reconstruction error is used as the anomaly score. This experiment checks whether deviations from the learned clean PPG distribution receive larger reconstruction errors. We report AUROC, AUPRC, TPR@\(5\%\)FPR, and the Spearman correlation between scores computed on input signals and scores computed on reconstructions. 

\subsection{Empirical Results: Anomaly Detection Sanity Check}
\label{subsec:anomaly_results}
As a sanity check on whether reconstruction error reflects signal quality, anomalous windows are constructed by injecting four types of synthetic corruption into random subset of the $1{,}645$ clean test windows. Each corruption type is applied to an equal share of the subset: $136$ windows receive additive Gaussian noise (motion artifact), $136$ receive slow baseline wander (breathing artifact), $135$ receive amplitude clipping (sensor saturation), and $135$ receive partial flatline (signal loss), yielding $542$ anomalous windows in total. The remaining $1{,}645$ windows are kept clean and treated as normal, giving a final evaluation set of $2{,}187$ windows with a $75.2\%$ / $24.8\%$ normal/anomalous split. 

Two anomaly scores are computed from each encode-decode pass: the mean absolute error $\mathrm{MAE}(x_0, \hat{x}_0)$ and the correlation-based score $1 - \rho$, where $\rho$ is the Pearson correlation between input and reconstruction. No anomaly-specific training is performed; the model was trained exclusively on clean signals, so corrupted inputs that fall outside the learned latent distribution are expected to produce higher reconstruction error.

Results are mixed and corruption-dependent, as shown in Table~\ref{tab:anomaly}. Gaussian noise is detected perfectly (AUROC$=1.0$, MAE$_{\text{anomalous}}=0.593$ vs 
MAE$_{\text{normal}}=0.039$), because high-frequency random corruption falls entirely outside the smooth periodic structure of the learned latent space. Baseline wander is partially separated (AUROC$=0.847$). Amplitude clipping performs below chance (AUROC$=0.400$) because clipped signals retain regular periodicity and are reconstructed faithfully as low-amplitude PPG, so the model cannot distinguish them from genuinely low-amplitude normal signals. The overall AUROC of $0.739$ should be interpreted cautiously given the poor clipping result; the Spearman correlation $r=0.98$ between input and reconstruction anomaly scores confirms the pipeline preserves signal quality ordering, but this experiment is not intended as a standalone anomaly detection benchmark.

\begin{table}[t]
\footnotesize
\centering
\caption{Anomaly detection sanity check results using reconstruction error as score.}
\begin{tabular}{lccc}
\toprule
Corruption & AUROC & AUPRC & TPR@5\%FPR \\
\midrule
Noise    & 1.000 & 1.000 & 1.000 \\
Baseline & 0.847 & 0.700 & 0.632 \\
Flatline & 0.708 & 0.149 & 0.111 \\
Clip     & 0.400 & 0.081 & 0.059 \\
\midrule
Overall  & 0.739 & 0.672 & 0.452 \\
\bottomrule
\end{tabular}
\label{tab:anomaly}
\end{table}

\begin{figure}[t]
    \centering
    \includegraphics[width=\linewidth]
    {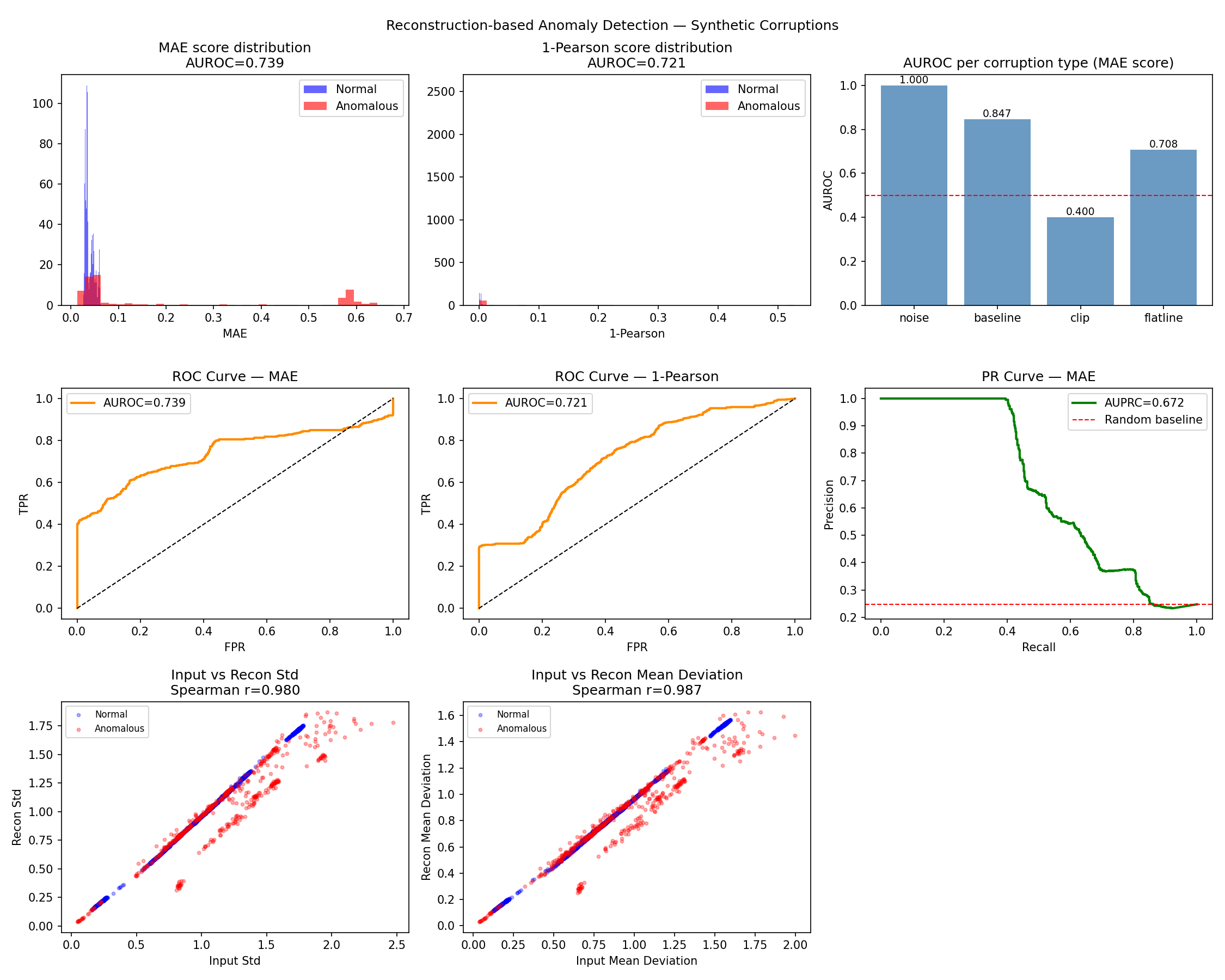}
    \caption{Score distributions, ROC curves, AUROC per corruption type, and input vs reconstruction scatter plots (Spearman $r=0.98$).}
    \label{fig:anomaly}
\end{figure}

\begin{figure}[H]
    \centering
    \includegraphics[width=\linewidth]
    {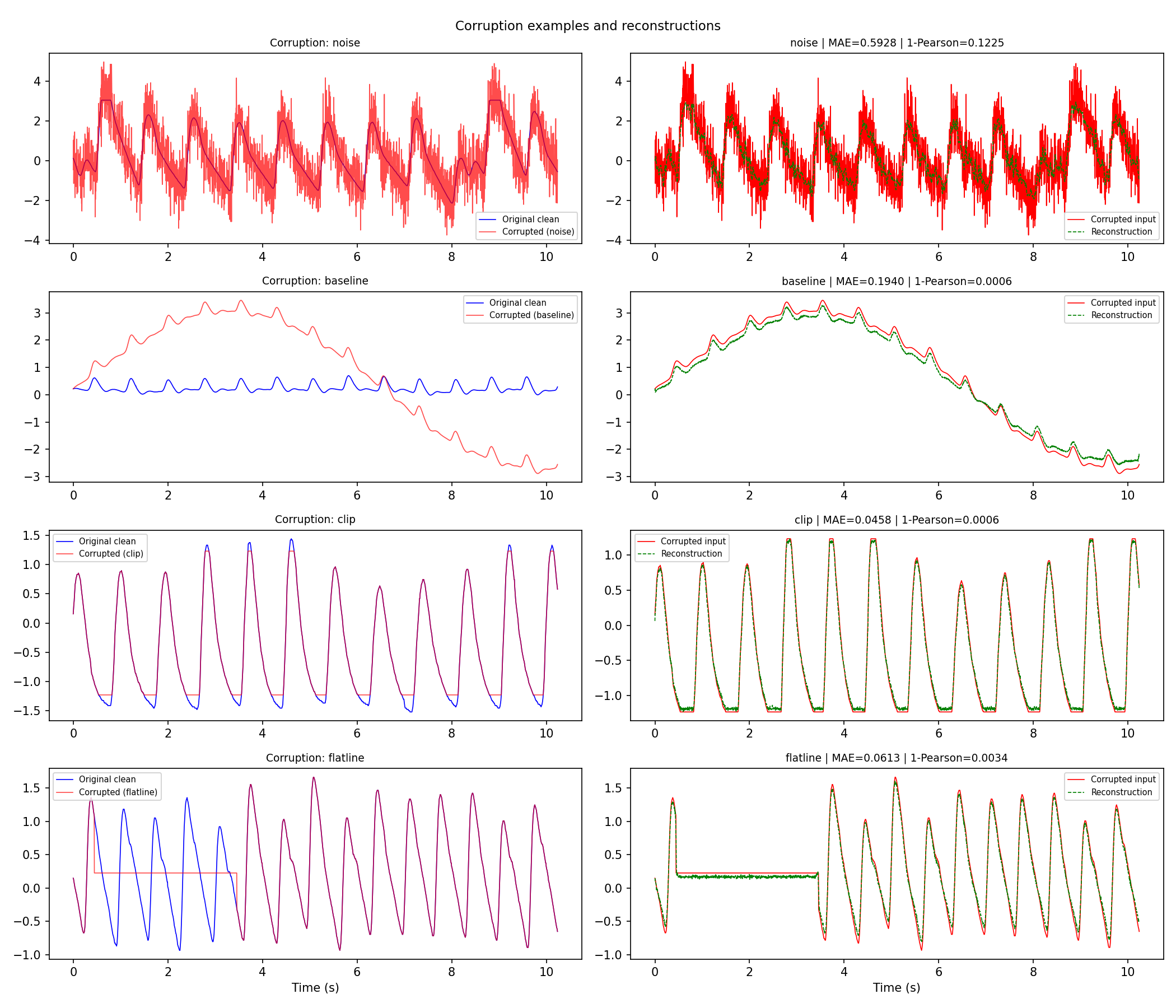}
    \caption{Example corruptions (left) and reconstructions (right). Noise produces high reconstruction error; clipping is reconstructed faithfully as low-amplitude PPG, explaining the poor detection rate.}
    \label{fig:corruption_examples}
\end{figure}

\section{Respiratory Rate Comparison}

Table~\ref{tab:rr_recon} compares CNN predictions on real and reconstructed PPG signals. The estimator produces similar predictions on real and reconstructed signals (mean $|\Delta| = 0.74$ bpm across recordings), indicating that the reconstruction pipeline does not substantially degrade the frozen estimator's predictions. Given the estimator's moderate accuracy on real signals, this should be interpreted as evidence that reconstruction preserves the signal features the estimator relies on, rather than as a claim of strong RR accuracy.

\begin{table}[t]
\footnotesize
\centering
\caption{RR estimation on real vs reconstructed PPG signals. Ground truth is derived from EtCO$_2$ capnography. ($|\Delta| = |\text{pred(real)} - \text{pred(recon)}|$)}
\begin{tabular}{lcccc}
\toprule
Recording & RR$_{\mathrm{EtCO}_2}$ & pred(real PPG) & pred(recon PPG) & $|\Delta|$ \\
\midrule
0312 & 11.84 & 12.37 & 12.98 & 0.61 \\
0313 & 14.53 & 16.23 & 15.82 & 0.41 \\
0322 & 18.00 & 17.56 & 15.58 & 1.98 \\
0328 & 12.45 & 12.82 & 12.99 & 0.17 \\
0331 & 14.71 & 16.76 & 16.83 & 0.07 \\
0333 & 12.14 & 12.62 & 13.97 & 1.35 \\
\bottomrule
\end{tabular}
\label{tab:rr_recon}
\end{table}

\section{Heart Rate Interpolation}

To assess latent space structure, two windows with heart rates of $65.9$ and $119.2$ bpm are encoded to $z_1$ and $z_2$, and signals are decoded from $z_\alpha = (1-\alpha)z_1 + \alpha z_2$ for $\alpha \in \{0, 0.25, 0.5, 0.75, 1.0\}$. The decoded heart rates transition as $65.9$, $65.9$, $100.1$, $119.1$, $119.2$ bpm, shown in Figure~\ref{fig:interpolation}, confirming that the latent space supports smooth generation across different heart rate regimes.

\begin{figure}[H]
    \centering
    \includegraphics[width=0.9\linewidth]
    {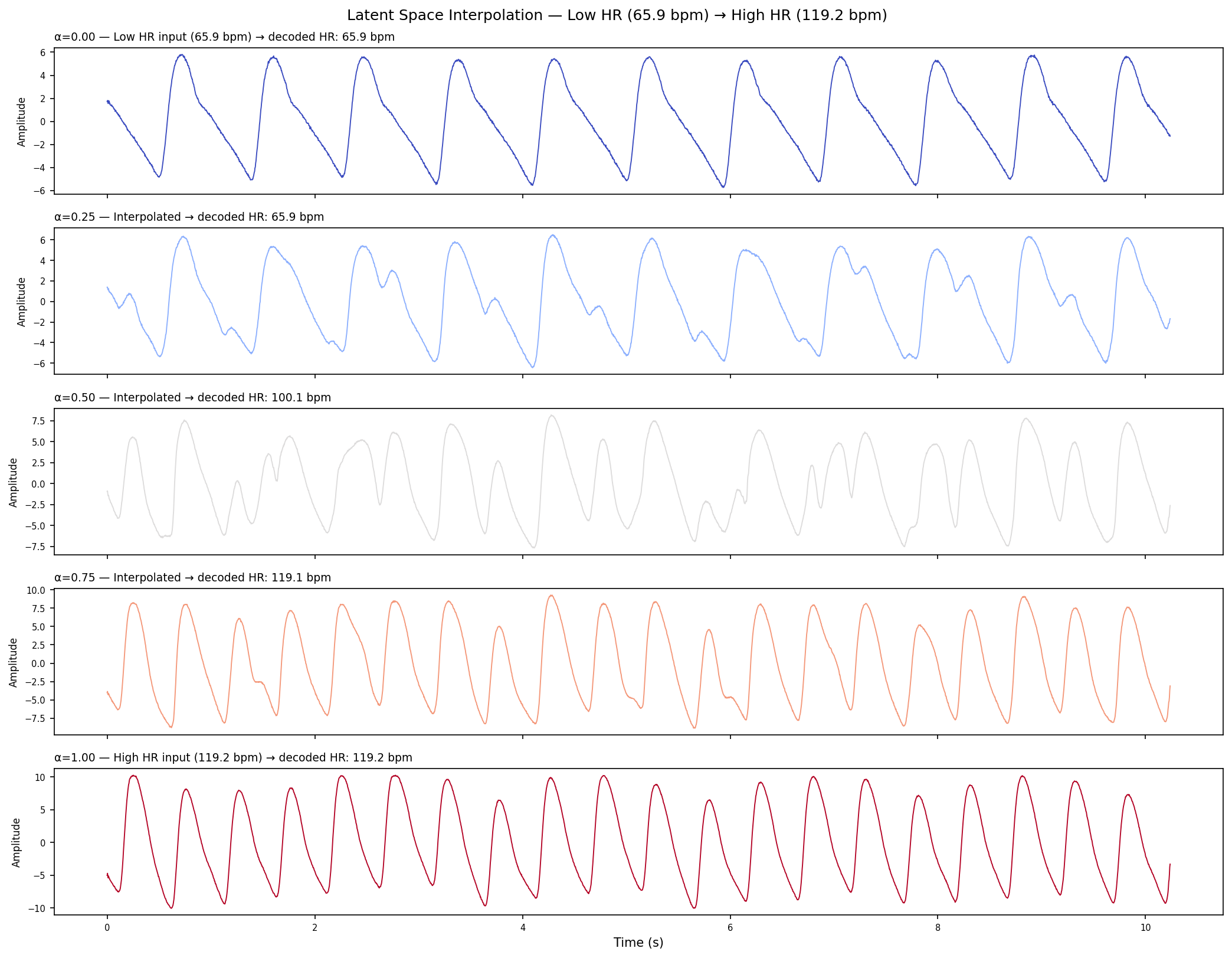}
    \caption{Decoded signals from latents interpolated between a low-HR ($65.9$ bpm) and high-HR ($119.2$ bpm) test window. Heart rate transitions smoothly from $65.9$ to $119.2$ bpm as $\alpha$ increases from $0$ to $1$.}
    \label{fig:interpolation}
\end{figure}
\end{document}